\documentclass[11pt]{article}%
\usepackage{amsfonts}
\usepackage{fancyhdr}
\usepackage[a4paper, top=2.5cm, bottom=2.5cm, left=2.2cm, right=2.2cm]%
{geometry}
\usepackage{amsmath}
\usepackage{amsthm}
\usepackage{changepage}
\usepackage{amssymb}
\usepackage{graphicx}%
\newtheorem{theorem}{Theorem}[section]
\newtheorem{corollary}[theorem]{Corollary}
\newtheorem{lemma}[theorem]{Lemma}
\newtheorem{proposition}[theorem]{Proposition}
\newtheorem{claim}[theorem]{Claim}
\newtheorem{observation}[theorem]{Observation}
\newtheorem{prop}[theorem]{Proposition}
\theoremstyle{definition} \newtheorem{definition}[theorem]{Definition}
\newcommand{\comm}[1]{}

\usepackage{todonotes}
\newcounter{todocounter}

\newcommand{\notshow}[1]{}

\usepackage{hyperref}
\hypersetup{colorlinks=true, citecolor=blue, linkcolor=black}

\begin{document}

\title{New Query Lower Bounds for Submodular Function Minimization} 
\author{Andrei Graur \and Tristan Pollner \and Vidhya Ramaswamy \and S. Matthew Weinberg\thanks{Supported by NSF CCF-1717899.}} 
\date{\today}
\maketitle
\begin{abstract} We consider submodular function minimization in the oracle model: given black-box access to a submodular set function $f:2^{[n]}\rightarrow \mathbb{R}$, find an element of $\arg\min_S \{f(S)\}$ using as few queries to $f(\cdot)$ as possible. State-of-the-art algorithms succeed with $\tilde{O}(n^2)$ queries~\cite{LeeSW15}, yet the best-known lower bound has never been improved beyond $n$~\cite{Harvey08}. 

We provide a query lower bound of $2n$ for submodular function minimization, a $3n/2-2$ query lower bound for the non-trivial minimizer of a symmetric submodular function, and a $\binom{n}{2}$ query lower bound for the non-trivial minimizer of an asymmetric submodular function. 

Our $3n/2-2$ lower bound results from a connection between SFM lower bounds and a novel concept we term the \emph{cut dimension} of a graph. Interestingly, this yields a $3n/2-2$ cut-query lower bound for finding the global mincut in an undirected, weighted graph, but we also prove it cannot yield a lower bound better than $n+1$ for $s$-$t$ mincut, even in a directed, weighted graph.

%The 3n/2 lower bound should be for undirected/weighted because our perturbation makes it weighted.
\end{abstract}
\newpage
\section{Introduction}
Submodular function minimization (SFM) is a classic algorithmic problem with numerous applications (e.g.~\cite{BoykovVZ01,KohliKT09,KohliT10,LinB11}): given black-box access to a submodular\footnote{$f(\cdot)$ is submodular if $f(X \cup Y) + f(X \cap Y) \leq f(X) + f(Y)$ for all sets $X, Y$. This is equivalent to $f(S \cup T \cup \{i\}) - f(S \cup T) \leq f(S \cup \{i\}) - f(S)$ for all $S,T,i$ (called \emph{diminishing marginal returns}).} function $f:2^{[n]} \rightarrow \mathbb{R}$, find an element of $\arg\min\{f(S)\}$. Due to its ubiquity within TCS and without, the problem has received substantial attention over the past four decades within various communities. Seminal work of Gr{\" o}tschel, Lovasz, and Schrijver first established that a minimizer can be found in poly-time~\cite{GrotschelLS81}, and after a long series of improvements the state-of-the-art now requires $\tilde{O}(n^2)$ value queries to $f(\cdot)$ and $\tilde{O}(n^3)$ additional overhead.

Despite remarkable progress on the algorithmic front, shockingly few \emph{lower bounds} on submodular function minimization are known. It is perhaps unsurprising that computational lower bounds are elusive, but \emph{even query lower bounds are virtually non-existent}. Indeed, state-of-the-art query lower bounds for SFM have remained stagnant at exactly $n$ for the past decade~\cite{Harvey08}. Our main results are new query lower bounds for three variants of SFM. We briefly provide the formal problem statements and our main results below, followed by an overview of context and related work.

\begin{definition}[Query Complexity of Submodular Function Minimization] Given as input black-box access to a submodular function $f(\cdot)$ over $n$ elements, output an element of $\arg\min_S\{f(S)\}$, along with $\min_S\{f(S)\}$. The \emph{query complexity} of SFM is equal to the minimum $q(\cdot)$ such that a deterministic algorithm solves SFM on all instances of $n$ elements with at most $q(n)$ queries. 
\begin{itemize}
\item If $f(\cdot)$ is further assumed \emph{symmetric}, i.e. $f(S) = f([n]\setminus S)$ for all $S$, this is \emph{Symmetric SFM}.
\item If we ask for $\arg\min_{S \notin \{\emptyset, [n]\}}\{f(S)\}$, this is \emph{Non-Trivial SFM} (we will also use both qualifiers).\footnote{Indeed, note that Symmetric SFM (without the Non-Trivial qualifier) is trivial, as $\emptyset$ (or $[n]$) is always a solution.}
\end{itemize}
\end{definition}

As a representative problem to have in mind, imagine a graph on $n$ nodes with positive edge weights and define $f(S)$ to be the weight of all edges leaving set $S$ (the value of cut $S$). Then $f(\cdot)$ is submodular, and non-trivial symmetric SFM would count the number of \emph{cut queries} needed to find the mincut. If you seek the minimum $s$-$t$ cut (and adjust notation so that $f(S)$ is equal to the weight of all edges leaving $S \cup \{s\}$), then this is standard SFM (because it is valid to output $\emptyset$, which implies that the mincut is $s$). If you seek the global mincut in a directed graph, then this is an instance of Non-Trivial SFM (because it is now invalid to output $\emptyset$). If you seek the global mincut in an undirected graph, then this is an instance of Non-Trivial Symmetric SFM (because it doesn't matter which side of the cut is sending versus receiving). Our main results are below.

\begin{theorem}[Main Results]\label{thm:main} The following lower bound the query complexity of SFM:
\begin{itemize}
\item The query complexity of SFM is at least $2n$.
\item The query complexity of Non-Trivial Symmetric SFM is at least $3n/2-2$.
\item The query complexity of Non-Trivial SFM is at least $\binom{n}{2}$.
\end{itemize}

\end{theorem}

\subsection{New Technique: The Cut Dimension}
Our SFM and Non-Trivial SFM lower bounds are direct constructions, and we defer all related intuition and technical details to the corresponding sections. Our Non-Trivial Symmetric SFM lower bound, however, derives from a new framework based on the \emph{cut dimension} of graphs. 

\begin{definition}[Global Cut Dimension, special case of Definition~\ref{def:dim}] Let $G$ be a directed graph with $m$ edges, and let $S$ be a subset of nodes. Define $\vec{v}^S$ be the vector in $\mathbb{R}^m$ with $v^S_e = 1$ iff the edge $e$ has left endpoint in $S$ and right endpoint not in $S$ (and $v^S_e = 0$ otherwise). Then the \emph{cut dimension} of $G$ is the dimension of $\text{span}(\{\vec{v}^S,\text{ $S$ is a global mincut}\})$. We also consider the following variants:
\begin{itemize}
\item If $G$ is undirected, then $v^S_e = 1$ iff the edge $e$ has one endpoint in $S$ and the other not in $S$.
\item If we seek the min $s$-$t$ cut, then $v^S_e = 1$ iff the edge $e$ has left endpoint in $S \cup \{s\}$ and right endpoint not in $S \cup \{s\}$. We also take the dimension over min $s$-$t$ cuts instead of global mincuts. In this case, we call this the $s$-$t$ Cut Dimension.
\end{itemize}
\end{definition}

Our main result concerning the cut dimension connects it to SFM lower bounds:

\begin{theorem}[Special case of Theorem~\ref{thm:equiv}]\label{thm:specdim} If an undirected graph exists with Global Cut Dimension $d$, then the query complexity of Non-Trivial Symmetric SFM is at least $d$.

If a graph exists with $s$-$t$ Cut Dimension $d$, then the query complexity of SFM is at least $d$. 

If a directed graph exists with Global Cut Dimension $d$, then the query complexity of Non-Trivial SFM is at least $d$.
\end{theorem}

Interestingly, we also establish that the Cut Dimension is a \emph{equivalent} to the best achievable lower bounds based on graphs via a canonical perturbation approach. Our $3n/2-2$ lower bound for Non-Trivial Symmetric SFM follows immediately from Theorem~\ref{thm:specdim} and the construction of an undirected graph with Global Cut Dimension $3n/2-2$. We also establish that every graph has $s$-$t$ Cut Dimension at most $n+1$, meaning this approach is useful for Non-Trivial SFM but not SFM. 

\subsection{Related Work}\label{sec:related}
The first poly-time (and strongly poly-time) algorithms for SFM were given by~\cite{GrotschelLS81} using the Lovasz extension and the Ellipsoid algorithm~\cite{Khachiyan79}. A substantial series of improvements followed over the subsequent four decades~\cite{Cunningham85,Schrijver00,FleischerI00,IwataFF01,Iwata03,Vygen03,Orlin07,IwataO09} The state-of-the-art is an $\tilde{O}(n^2)$ upper bound on the query complexity of SFM~\cite{LeeSW15}, an $O(n^3)$ upper bound on the query complexity of Non-Trivial Symmetric SFM~\cite{Queyranne98}, and an $\tilde{O}(n^3)$ upper bound on the query complexity of Non-Trivial SFM~\cite{LeeSW15}.\footnote{The final bound follows by a reduction from Non-Trivial SFM to SFM incurring a blowup of $2n$ (for all elements $i$, run SFM only over sets containing $i$ and not containing $i+1$, and then only over sets containing $i$ and not $i-1$).} 

Despite this substantial progress on upper bounds, the only unconditional query lower bound is just $n$ (which is surprisingly non-trivial to establish)~\cite{Harvey08}.~\cite{ChakrabartyLSW16} give a construction which requires $\Omega(n)$ \emph{queries to the Lovasz extension} (if the algorithm can only query the Lovasz extension), but one can find the minimizer in their construction by simply querying all $n$ singletons. 

Our Non-Trivial Symmetric SFM lower bound uses the cut function in graphs. Recent work of~\cite{RubinsteinSW18} establishes that this particular instance of Non-Trivial Symmetric SFM (in unweighted graphs) can be solved by a randomized algorithm in $\tilde{O}(n)$ queries, but our techniques are unrelated.

\subsection{Roadmap}
Section~\ref{sec:main} provides our $2n$ lower bound for SFM, which is a direct construction. Section~\ref{sec:symmetric} proves (a generalization of) Theorem~\ref{thm:specdim} and provides a graph with Global Cut Dimension $3n/2-2$, yielding our $3n/2-2$ lower bound for Non-Trivial Symmetric SFM. Section~\ref{sec:nontriv} provides our $\binom{n}{2}$ lower bound for Non-Trivial SFM, which is also a direct construction.

Appendix~\ref{app:auxiliary} contains auxiliary claims concerning cut queries in graphs (i.e., it is impossible to learn precisely a directed graph using cut queries, what can you learn?) which are not necessary for our lower bounds, but likely useful for future work. Appendix~\ref{app:proofs} contains one omitted proof.

\section{A $2n$ Query Lower Bound for SFM}\label{sec:main}
This section proves our lower bound on SFM.

\begin{theorem} 
\label{thm:2nbound}
The query complexity of SFM is at least $2n$.
\end{theorem} 

Let us first provide intuition for our construction. We start with an arbitrary permutation $\sigma$ on $[n]$, and define the \emph{important sets} $R_i$ for $0 \le i \le n$ by $R_i := \{\sigma(1), \sigma(2), ..., \sigma(i)\}$ for each $i \in [n] \cup \{0\}$. Observe that there is exactly one important set of each size from $0$ to $n$. These important sets will be the potential minimizers. Intuitively, we will define our function such that: (a) any algorithm must query at least $n-1$ unimportant sets to learn the important sets, and (b) any algorithm must query all $n+1$ important sets to learn the minimizer. In detail:
\begin{itemize}
\item Let $\sigma$ be an arbitrary permutation on $[n]$.
\item Define $R_i:= \{\sigma(1),\ldots, \sigma(i)\}$ for each $i \in [n] \cup \{0\}$.
\item For each $i \in [n]\cup \{0\}$, let $c_i \in \{0,1\}$. 
\item Define the function $f_\sigma^{\vec{c}}(\cdot)$ such that (below, $j(S)$ denotes the maximum $j$ such that $R_j \subseteq S$):
$$f_\sigma^{\vec{c}} (S) =  \begin{cases}
      -c_i & \text{if } S = R_i \text{ for some } 0 \le i \le n \\
      (|S|-j(S)) \cdot (n + 2 - j(S))& \text{else} 
\end{cases}$$
\end{itemize}

That is, $f_\sigma^{\vec{c}}(\cdot)$ is defined to be non-negative on the unimportant sets, and non-positive on the important sets. Intuitively, queries to unimportant sets give information regarding $\sigma$, and queries to important sets give information regarding $\vec{c}$. It is not obvious, but straight-forward to establish that $f_\sigma^{\vec{c}}(\cdot)$ is submodular for all $\sigma, \vec{c}$. The proof of Lemma~\ref{lem:submodSFM} appears in Appendix~\ref{app:proofs}.

\begin{lemma}\label{lem:submodSFM} For all $\sigma, \vec{c}$, $f_\sigma^{\vec{c}}$ is submodular.
\end{lemma} 

We now provide a complete proof that deterministic algorithms must make $2n$ queries for functions of the form $f_\sigma^{\vec{c}}(\cdot)$. We define an adversary which adaptively sets $\sigma, \vec{c}$ as queries are made:

\begin{itemize}
\item Initialize $\sigma, \vec{c}$ to be undefined.
\item Let $i$ denote the maximum $j$ such that $\sigma(j)$ is defined (so initially $i=0$). 
\item When a new query, $S$, is made:
\begin{enumerate}
\item If $S= R_j$, for some $j \leq i$, answer $0$ and set $c_j = 0$. Call this an \emph{important query}.
\item If $R_i \not\subseteq S$, $j(S)$ is defined. Answer $(|S|-j(S))\cdot (n+2-j(S))$. Call this a \emph{useless query}. 
\item If $R_i \subset S$, $j(S)$ is not yet defined. Pick any $j \notin S$ (such a $j$ must exist as $S \neq R_n$) and set $\sigma(i+1) = j$.\footnote{Observe also that $j \neq \sigma(\ell)$ for any $\ell \leq i$, as $R_i \subset S$.} Now $j(S):=i$, so answer $(|S|-i)\cdot (n+2-i)$. Call this a \emph{decoy query}.
\end{enumerate}
\item If the algorithm terminates after $n+1$ (distinct) important queries, $\sigma$ and $\vec{c}$ are fully defined.
\item If the algorithm has made fewer than $n+1$ (distinct) important queries, let $x$ denote the algorithm's guess for the minimum value.
\begin{enumerate}
\item If $x = 0$, set all undefined $c_i:=1$ and complete $\sigma$ arbitrarily (if necessary).
\item If $x = -1$, set all undefined $c_i:=0$ and complete $\sigma$ arbitrarily (if necessary).
\item If $x \notin \{0,-1\}$, complete $\vec{c},\sigma$ arbitrarily.
\end{enumerate}
\end{itemize}

Theorem~\ref{thm:2nbound} will follow by proving that the above adversary is consistent and that the adversary has the power to make multiple (distinct) minima unless the algorithm has made at least $n-1$ decoy queries and $n+1$ important queries.

\begin{observation}\label{obs:valid} The adversary answers all queries in a way that is consistent with some $f_\sigma^{\vec{c}}(\cdot)$.
\end{observation}
\begin{proof}
Observe that the adversary answers all queries to $R_j$ with $0$, so this is always consistent. Further observe that whenever an unimportant set is queried, either the answer is already determined by $\sigma$ (and therefore consistent), or one new output of $\sigma$ is fixed so that the answer is now determined by $\sigma$ (and therefore consistent now and forever). The precise definition of $f_\sigma^{\vec{c}}(\cdot)$ is important for the final claim: as soon as we know that $\sigma(i+1) \notin S$, this fixes the value of $f_\sigma^{\vec{c}}(\cdot)$. 

Finally, observe that the completion step is also consistent with all previous queries, as they are completely defined by the partial definition of $\sigma, \vec{c}$.
\end{proof}

\begin{lemma}\label{lem:decoy}
Algorithms cannnot make $n+1$ distinct important queries without $n-1$ decoy queries.
\end{lemma}
\begin{proof}
Observe that each decoy query increases $i$ by one. Observe that the only distinct important queries that can be made are $\emptyset, [n]$, and $R_1,\ldots, R_i$, for a total of $i+2$. If $i < n-1$, then the distinct possible important queries are also $< n+1$.
\end{proof}

\begin{lemma}\label{lem:important}
Any algorithm making $<n+1$ distinct important queries is wrong.
\end{lemma}
\begin{proof}
If the guess is $\notin\{0,-1\}$, then the guess is clearly wrong. If the guess is $0$, then the completion step makes it so that the minimum is $-1$, so the guess is wrong. If the guess is $-1$, then the completion step makes it so that the minimum is $0$, so the guess is wrong.
\end{proof}

\begin{proof}[Proof of Theorem~\ref{thm:2nbound}] Lemmas~\ref{lem:important} and~\ref{lem:decoy} together assert that the algorithm must make $n-1$ decoy queries and $n+1$ important queries in order to correctly solve SFM on instances of the form $f_\sigma^{\vec{c}}(\cdot)$ against the prescribed adversary. Therefore, a total of $2n$ queries must be made.
\end{proof}

We conclude this section by noting that our construction witnesses a lower bound of exactly $2n$ (and no better).

\begin{proposition}An SFM algorithm exists making $2n$ queries for any function of the form $f_\sigma^{\vec{c}}(\cdot)$.
\end{proposition}
\begin{proof}
First, query the $n-1$ sets $[n]\setminus \{i\}$ for all $i \neq 1$. Observe that $f_\sigma^{\vec{c}}([n]\setminus \{i\}) \leq 0$ if and only if $\sigma(n) = i$. Similarly, as $[n]\setminus \{i\}$ is missing only a single element ($i$), this means that if $\sigma(n) \neq i$ then $f_\sigma^{\vec{c}}([n]\setminus \{i\}) = (n - \sigma^{-1}(i))\cdot (n+3 - \sigma^{-1}(i))$. So the query to $[n]\setminus \{i\}$ reveals $\sigma^{-1}(i)$, for any $i$. Therefore, these $n-1$ queries completely reveal $\sigma$ (because $\sigma$ is a permutation).

After $\sigma$ is fully revealed, simply query the $n+1$ important sets to find the minimizer.
\end{proof}

\section{A $3n/2-2$ Query Lower Bound for Non-Trivial Symmetric SFM}\label{sec:symmetric}
We begin this section by providing a generalization the cut dimension, first by providing a class of submodular functions which generalize mincuts in graphs.

\subsection{Defining the Generalized Cut Dimension}
 Consider a ground set of $n$ elements, and a disjoint set of $m$ \emph{hyperedges}. We associate with each $S \subseteq [n]$ (including $S = \emptyset$) a set $h(S) \subseteq [m]$ of hyperedges that are active for $S$. For example, to capture mincuts in an undirected graph we might have the hyperedges simply be the edges of that graph, and $h(S)$ would denote the edges with one endpoint in $S$ and the other not in $S$. 

To each $i \in [m]$, associate a non-negative weight $w_i$, and define the function $f(\cdot)$ so that $f(S):= \sum_{i \in h(S)} w_i$. If the active sets $h(\cdot)$ satisfy the following inequality, then it is easy to see that $f(\cdot)$ is submodular (below, $X \overline{\cup}Y$ denotes the multiset union of $X$ and $Y$, which contains two copies of every element in $X \cap Y$):
\begin{equation*}
h(S \cap T) \overline{\cup} h(S \cup T) \subseteq h(S) \overline{\cup} h(T).
\end{equation*}

We call such functions \emph{weight-based}. It is easy to see that cuts in graphs or hypergraphs are weight-based. For such functions, there is a meaningful notion of ``dimension" associated with the set of minimizers. For every $S \subseteq [n]$, define the vector $\vec{v}^S \in \mathbb{R}^m$ so that $v^S_i=1$ if and only if $i \in h(S)$ \emph{and $w_i > 0$}, and $v^S_i = 0$ otherwise. For a set $\mathcal{S}$ of subsets of $[n]$, let $\text{dim}(\mathcal{S})$ denote the dimension of the span (over $\mathbb{R}^m$) of the vectors $\{\vec{v}^S\}_{S \in \mathcal{S}}$. We now define the Generalized Cut Dimension:

\begin{definition}[Generalized Cut Dimension]\label{def:dim} Let $f(\cdot)$ be weight-based. Then the \emph{Generalized Cut Dimension} of $f(\cdot)$ is equal to $\text{dim}(\arg\min_S\{f(S)\})$. The \emph{Generalized Non-Trivial Cut Dimension} of $f(\cdot)$ is $\text{dim}(\arg\min_{S \notin \{\emptyset,[n]\}}\{f(S)\}\})$.
\end{definition}

We will call a weight-based function symmetric if $h(S) = h([n]\setminus S)$ --- it is clear that the resulting submodular function is symmetric.

\subsection{Connecting Generalized Cut Dimension to Query Complexity}
In this section, we establish the \emph{equivalence} of Generalized Cut Dimension to a canonical ``perturbation'' approach for lower bounding the query complexity. 

\begin{definition}[Perturbation Bound] Starting from a (symmetric, if desired) weight-based submodular function $f(\cdot)$ with weights $\vec{w}$ and (non-trivial, if desired) minimizers $\mathcal{M}_f$, pick a sufficiently small $\varepsilon>0$ so that every $\vec{w}'$ with $w'_i \in [(1-\varepsilon)w_i,(1+\varepsilon)w_i]$ induces a (symmetric, if desired) weight-based submodular function $g(\cdot)$ with (non-trivial, if desired) minimizers $\mathcal{M}_g \subseteq \mathcal{M}_f$. Let $\mathcal{G}(f)$ denote the set of all (symmetric, if desired) weight-based functions with $w'_i \in [(1-\varepsilon)w_i,(1+\varepsilon)w_i]$. 

If it is the case that, for any set of $q-1$ queries to $f$, there exists a $g \in \mathcal{G}(f)$ consistent with those queries such that $\min_S\{g(S)\} \neq \min_S\{f(S)\}$, we say that $f$ witnesses a (Symmetric) \emph{Perturbation Bound} of $q$ (if there is a $g(\cdot) \in \mathcal{G}(f)$ with $\min_{S \notin \{\emptyset, [n]\}}\{g(S)\} \neq \min_{S \notin \{\emptyset, [n]\}}\{f(S)\}$, we say that $f$ witnesses a Non-Trivial Perturbation Bound of $q$).

We refer to \emph{the} Perturbation Bound of $f$ as the maximum possible $q$ such that $f$ witnesses a Peturbation Bound of $q$ (and similarly can define \emph{the} Non-Trivial Perturbation Bound). 
\end{definition}

Intuitively, the perturbation bound captures the following natural way to obtain query lower bounds: start from some function $f(\cdot)$ with minimizers $\mathcal{M}_f$. There is some non-zero gap $\delta$ between the minimizers and the rest, so there exists a sufficiently small $\varepsilon$ such that perturbing weights by $\varepsilon$ can tie-break among minimizers, but not yield a new minimizer. 

\begin{observation}If there exists an $f(\cdot)$ witnessing a (Symmetric, Non-Trivial) Perturbation Bound of $q$, then the query complexity of (Symmetric, Non-Trivial) SFM is at least $q$.
\end{observation}
\begin{proof}
Assume for contradiction that an algorithm correctly outputs the minimum value, $x$, after $q-1$ queries that are consistent with $f(\cdot)$. If $x \neq \min_S \{f(S)\}$, then all queries are consistent with $f(\cdot)$, so the algorithm could be wrong because the function is $f(\cdot)$. if $x = \min_S\{f(S)\}$, then because $f(\cdot)$ witnesses a Perturbation Bound of $q$, there exists a $g(\cdot)$ consistent with all $q-1$ queries with $\min_S \{g(S)\} \neq x$, so the algorithm could be wrong because the function is $g(\cdot)$.

The same proof holds verbatim if $f(\cdot)$ is assumed to be symmetric, or if we replace absolute minimizers with non-trivial minimizers.
\end{proof}

We now establish that the perturbation bound approach yields exactly the same lower bound as the generalized cut dimension. Our proof will make use of the following observation.

\begin{observation}\label{obs:dotprod} For any $g \in \mathcal{G}(f)$, $g(S)=\vec{v}^S \cdot \vec{w}'$ (where $\vec{w}'$ is the weight vector defining $g(\cdot)$). 
\end{observation}
\begin{proof}
First, recall that $g(S):= \sum_{i \in h(S)} w'_i$. Recall further that $v_i^S = 1$ whenever $w_i \neq 0$ and $i \in h(S)$. Importantly, note that $w_i=0 \Leftrightarrow w'_i = 0$ for all $g(\cdot) \in \mathcal{G}(f)$, so in fact $v_i^S = 1$ whenever $w'_i \neq 0$ and $i \in h(S)$ (and $0$ otherwise). This immediately implies that $\vec{v}^S \cdot \vec{w}' = \sum_{i \in h(S)} w'_i = g(S)$.
\end{proof}

\begin{theorem}\label{thm:equiv} Let $f(\cdot)$ be (symmetric) weight-based. Then the (Symmetric, Non-Trivial) Perturbation Bound of $f(\cdot)$ is exactly equal to the Generalized (Non-Trivial) Cut Dimension of $f(\cdot)$.
\end{theorem}
\begin{proof}
We break the proof down into two lemmas, one establishing that the Perturbation Bound is at most the Generalized Cut Dimension, and one establishing that the Perturbation Bound is at least the Generalized Cut Dimension. We first establish the easy direction, that if $f(\cdot)$ has Generalized (Non-Trivial) Cut Dimension $d$, it witnesses a (Symmetric, Non-Trivial) Perturbation Bound of at most $d$. For simplicity of notation throughout the proof, we explicitly prove the standard case, but the claims for Symmetric and Non-Trivial follow verbatim.

\begin{lemma}\label{lem:oneway} For all (Symmetric) weight-based $f(\cdot)$, the (Symmetric, Non-Trivial) Perturbation Bound is at most the Generalized (Non-Trivial) Cut Dimension.
\end{lemma}
\begin{proof}
Say that $f(\cdot)$ has Generalized Cut Dimension $d$, and let $S_1,\ldots, S_d$ be such that $v^{S_1},\ldots, v^{S_d}$ form a basis for the span of $\{\vec{v}^S\}_{S \in \mathcal{M}_f}$. We claim that queries to $S_1,\ldots, S_d$ completely determine $g(S) = f(S)$ for all $S \in \mathcal{M}_f$ and all $g(\cdot) \in \mathcal{G}(f)$. If true, this establishes a set of $q$ queries for which there does not exist a $g(\cdot) \in \mathcal{G}(f)$ consistent with these queries for which $\min_S\{g(S)\} \neq \min_S\{f(S)\}$ (and therefore the Perturbation Bound for $f(\cdot)$ is at most $d$). 

Consider any $S \in \mathcal{M}_f$. Then we know by definition of the Generalized Cut Dimension that $\vec{v}^S = \sum_i c_i \vec{v}^{S_i}$ for some $c_1,\ldots, c_d \in \mathbb{R}$. We claim that this implies that $g(S)=\sum_i c_i g(S_i)$ \emph{for any $g(\cdot) \in \mathcal{G}(f)$}. This follows from the following equalities, which make use of Observation~\ref{obs:dotprod}.

\begin{align*}
g(S) &= \vec{v}^S \cdot \vec{w}'\\
&= \sum_i c_i \vec{v}^{S_i} \cdot \vec{w}'\\
&= \sum_i c_i \cdot g(S_i).
\end{align*}

Therefore, if we query $S_1,\ldots, S_d$ and learn that $g(S_i)=f(S_i)$, this fully determines $g(S) = f(S)$ for all $S \in \mathcal{M}_f$, and therefore establishes that the Perturbation Bound for $f(\cdot)$ is at most $d$.
\end{proof}

We now show the hard direction: Perturbation Bound is at least the Generalized Cut Dimension.

\begin{lemma}\label{lem:hardway}  Any (Symmetric) weight-based $f(\cdot)$ witnesses a (Symmetric, Non-Trivial) Perturbation Bound of $d$, the (Non-Trivial) Generalized Cut Dimension.
\end{lemma}
\begin{proof}
Let $T_1,\ldots, T_{d-1}$ be any $d-1$ sets queried. Let $X$ denote the subspace of vectors $\vec{y}$ which satisfy the linear equations $\vec{v}^{T_i} \cdot \vec{y} = 0$ for all $1 \le i \le d-1$. Observe that $X$ has dimension at least $m-d+1$. Let $Y$ denote the subspace of vectors spanned by $\{\vec{v}^S\}_{S \in \mathcal{M}_f}$. Observe that $Y$ has dimension $d$.

The dimension of $X$ ($\geq m-d+1$) and the dimension of $Y$ ($d$) sum to $> m$. This means that there exists a non-zero vector, $\vec{z} \in X \cap Y$.\footnote{To see why this is the case, write a basis $\mathcal{B}_X = \{v_1, v_2, \ldots, v_{m-q} \}$ of $X$ and a basis $\mathcal{B}_Y = \{w_1, w_2, \ldots, w_d \}$ of $Y$. If $\mathcal{B}_X \cap \mathcal{B}_Y$ is non-empty then we are of course done, otherwise $\mathcal{B}_X \cup \mathcal{B}_Y$ is a set of strictly more than $m$ vectors in $\mathbb{R}^m$. Hence they must be linearly dependent, implying we can write $\alpha_1 v_1 + \ldots \alpha_{m-q} v_{m-q} + \beta_1 w_1 + \ldots + \beta_d w_d = 0$ for some coefficients $\{ \alpha _i \}$, $\{ \beta_j \}$ that are not all zero. Then note that $\beta_1 w_1 + \ldots \beta_d w_d$ is clearly in both $X$ and $Y$, and cannot be zero as otherwise all coefficients would be zero (because both $\mathcal{B}_X$ and $\mathcal{B}_Y$ are linearly independent).} Because $\vec{z} \in X$, we can add $\varepsilon \vec{z}$ to $\vec{w}$ for any $\varepsilon$ and arrive at a $\vec{w}'$ which is consistent with the queries so far. Because $\vec{z} \in Y$, we must have $z_i = 0$ whenever $w_i = 0$ (because all $\vec{v}^S$ have $v_i = 0$ when $w_i = 0$). Therefore, there exists a sufficiently small $\varepsilon$ such that $\vec{w} + \varepsilon \vec{z}$ results in a $g(\cdot)$ which is consistent with all $d-1$ queries so far, and is in $\mathcal{G}(f)$, and also $\vec{w} - \varepsilon \vec{z}$ results in such a $g(\cdot)$ as well.

Consider now $\vec{z} \cdot \vec{v}^S$ for any $S \in \mathcal{M}_f$. If $\vec{z} \cdot \vec{v}^S > 0$, then when $\vec{w}' := \vec{w} - \varepsilon \vec{z}$, we have $g(S) < f(S)$, and therefore $\min_S \{g(S)\} < \min_S \{f(S)\}$, meaning that we have found the desired Perturbation Bound $g(\cdot)$ for these $d-1$ queries. If $\vec{z} \cdot \vec{v}^S < 0$ we can instead use $\vec{w}':= \vec{w} + \varepsilon \vec{z}$. So if these $d-1$ queries have no witness, it must be that $\vec{z} \cdot \vec{v}^S = 0$ for all $S \in \mathcal{M}_f$. In particular that this holds for the basis $\vec{v}^{S_1},\ldots, \vec{v}^{S_d}$ of $Y$. To summarize this paragraph: unless $\vec{v}^{S_i} \cdot \vec{z} = 0$ for all $i$ (note that these $S_i$ were not necessarily queried), then these $d-1$ queries have a witness for the Perturbation Bound. 

%Consider now $\vec{z} \cdot \vec{v}^S$ for any $S \in \mathcal{M}_f$, and assume for contradiction that $\vec{z} \cdot \vec{v}^S \neq 0$. If $\vec{z} \cdot \vec{v}^S > 0$, then when $\vec{w}' := \vec{w} - \varepsilon \vec{z}$, we have $g(S) < f(S)$, and therefore $\min_S \{g(S)\} < \min_S \{f(S)\}$, meaning that we have found the desired Perturbation Bound $g(\cdot)$ for these $q$ queries. If $\vec{z} \cdot \vec{v}^S < 0$ we can instead use $\vec{w}':= \vec{w} + \varepsilon \vec{z}$. So if these $q$ queries have no witness, it must be that $\vec{z} \cdot \vec{v}^S = 0$ for all $S \in \mathcal{M}_f$, and in particular that this holds for the basis $\vec{v}^{S_1},\ldots, \vec{v}^{S_d}$ of $Y$. To summarize this paragraph: unless $\vec{v}^{S_i} \cdot \vec{z} = 0$ for all $i$ (note that these $S_i$ were not necessarily queried), then these $q$ queries have a witness for the Perturbation Bound. 

We will now establish that we can't have $\vec{z} \cdot \vec{v}^{S_i} = 0$ for all $i$, which will establish that in fact there is a witness for these $d-1$ queries (and all $d-1$ queries, since they were arbitrary). Consider that because $\vec{z} \in Y$, we can write $\vec{z} = \sum_i \beta_i \vec{v}^{S_i}$ for some $\vec{\beta}$ which is not $\vec{0}$. If $\vec{z} \cdot \vec{v}^{S_i} = 0$ for all $i$ we have

$$\sum_i \beta_i \vec{v}^{S_i} \cdot \vec{v}^{S_j} = 0, \quad \forall j.$$
Therefore, if we let $A$ denote the $d \times m$ matrix whose rows are the vectors $\vec{v}^{S_i}$, we get that:

$$(A \cdot A^T) \begin{bmatrix}
	\beta_1 \\
	\vdots \\
	\beta_d
	\end{bmatrix} = 0$$
 Because $\vec{v}^{S_1},\ldots, \vec{v}^{S_d}$ form a basis for $Y$ we know $A, A^T$, and $A \cdot A^T$ have rank $d$.\footnote{A short proof of this is through the singular value decomposition (SVD) of $A$. Write $A = U \Sigma V$ where $U \in \mathbb{R}^{d \times d}$, $\Sigma \in \mathbb{R}^{d \times m}$, and $V \in \mathbb{R}^{m \times m}$ (where $U$ and $V$ satisfy $UU^T = I$ and $VV^T = I$, and $\Sigma$ is diagonal with $d$ non-zero entries). Note then that $A \cdot A^T = U \Sigma V V^T \Sigma^T U^T = U(\Sigma \Sigma^T)U^T$. As $\Sigma\Sigma^T$ is diagonal of rank $d$ and $UU^T = I$, this is a singular value decomposition of $A\cdot A^T$, and directly implies that its rank is $d$.} But $\vec{\beta}$ is a non-zero vector in the kernel of the $d \times d$ matrix $A \cdot A^T$, which is a contradiction. Therefore, we must have $\vec{z} \cdot \vec{v}^{S_i} \neq 0$ for some $i$, implying that there exists the desired $g(\cdot)$ for any set of $d-1$ queries, and the Perturbation Bound is hence at least $d$.
\end{proof}

The proof of the theorem now follows directly from Lemmas~\ref{lem:oneway} and~\ref{lem:hardway}.
\end{proof}

Theorem~\ref{thm:equiv} lets us now restrict attention to the study of generalized cut dimension if we aim to prove lower bounds through the canonical perturbation approach. The subsequent sections establish that this is fruitful for symmetric, non-trivial SFM, but not for standard SFM.

\subsection{An Undirected Graph with Global Cut Dimension $3n/2-2$}
In this section, we provide an explicit undirected graph $G$ on $n$ vertices which has global cut dimension $3n/2-2$. This establishes the following theorem:

\begin{theorem}
The query complexity of Symmetric Non-Trivial SFM is at least $3n/2-2$.
\end{theorem}

\begin{proof}
First, let $n$ be odd and $n \geq 3$. Then $n = 2a+1$ for some $a \geq 1$, so label the vertices of $G$ as $\{v,w_1,w'_1,w_2,w'_2,\ldots, w_a,w'_a\}$. For edges (all undirected), put an edge between $v$ and all other nodes, and an edge between $w_i$ and $w'_i$ for all $i$ (and no other edges). It is easy to see that every cut in $G$ has value at least $2$, and that the mincuts indeed have value $2$. The mincuts either separate $w_i$ from the rest of the graph, $w'_i$ from the rest of the graph, or $\{w_i,w'_i\}$ from the rest of the graph. 

For any $i \in [a]$, let $i_1$, $i_2$, $i_3$ denote the three positions in indicator vectors corresponding to the three edges $(v,w_i)$, $(v,w'_i)$, $(w_i,w'_i)$. Restricted to these positions, the indicator vectors for the three minimum cuts $\{w_i \}$, $\{ w_i ' \}$, $\{w_i,w'_i\}$ are $(1,0,1)$, $(0,1,1)$, and $(1,1,0)$, which span a subspace of dimension three. As these three cuts have zeroes for all other entries, the full vectors also span a subspace of dimension three. For each $i \in [a]$, the set of three indices referenced above are distinct, which means that taking all these indicator vectors together has rank $3a$. 

As $n=2a+1$, $3a = 3(n-1)/2$. So the claim holds when $n$ is odd. If $n$ is even, use exactly the same construction on $n-1$ nodes, and connect the remaining node to $v$ with an edge of weight two. Now there is one additional mincut (separating the extra node from the rest), so the dimension is $3a+1 = 3(n-2)/2+1 = 3n/2-2$. 
\end{proof} 

\subsection{Generalized Cut Dimension is at most $n+1$}
The previous section establishes that the Non-Trivial Symmetric Cut Dimension can be much larger than $n$, which leads to novel lower bounds. In this section, we establish that this approach will not yield novel lower bounds for standard SFM (and by Theorem~\ref{thm:equiv}, neither will the canonical perturbation argument for weight-based functions).	

\begin{theorem} The Generalized Cut Dimension of any weight-based function is at most $n+1$.
\end{theorem}
\begin{proof}
First, recall that for any submodular $f(\cdot)$, the set of minimizes $\mathcal{M}_f$ is closed under union and intersection.\footnote{To see this, recall that $f(S \cup T) + f(S \cap T) \leq f(S) + f(T)$. If both of the sets on the RHS are minimizers, both sets on the LHS must be as well.} For each $i \in [n]$, define $S_i:= \cap_{S \in \mathcal{M}_f, i \in S}S$ (if there exists a minimizer containing $i$, otherwise let $S_i$ be null). If $S_i$ is not null, then $S_i \in \mathcal{M}_f$, because $\mathcal{M}_f$ is closed under intersection. Our goal will be to show that these $S_i$ (and $\emptyset$, if $\emptyset \in \mathcal{M}_f$) span $\mathcal{M}_f$ through a sequence of lemmas.

Define the \emph{base sets} $\mathcal{B}$ of $\mathcal{M}_f$ to be the set of all non-null $S_i$, together with $\emptyset$ (if $\emptyset \in \mathcal{M}_f$). It is clear that $\mathcal{B}$ has size at most $n+1$. It is also the case that every minimizer $S \in \mathcal{M}_f$ can be written as the union of elements in $\mathcal{B}$:

\begin{lemma} For all $S \in \mathcal{M}_f$, $S= \cup_{i \in S} S_i$. \end{lemma} 

\begin{proof} For any $i \in S$, we know that $i \in S_i$ which means that $S \subseteq \cup_{i \in S} S_i$. We need only show that for any $i \in S$, $S_i \subseteq S$. This is true by definition of $S_i$ because $S$ is a minimizer containing $i$. \end{proof} 

We next show that the base sets ``cover'' $\mathcal{M}_f$ in the following sense. Say that a set $S$ is covered by $\mathcal{B}$ if either: (a) $S \in \mathcal{B}$, or (b) $S$ can be written as the union of two sets which are covered by $\mathcal{B}$.

\begin{lemma}\label{lem:covered}Every set in $\mathcal{M}_f$ is covered by $\mathcal{B}$.
\end{lemma}
\begin{proof}
Assume for contradiction that there is some $S \in \mathcal{M}_f$ which is not covered by $\mathcal{B}$. Take the $S$ which minimizes $|S|$. Then clearly $S \notin \mathcal{B}$. So pick an arbitrary $i \in S$ and we can write $S = \cup_{j \in S} S_j = S_i \cup (\cup_{j \in S \setminus S_i} S_j)$. $S_i$ is clearly non-empty, and also because $S \notin \mathcal{B}$, it is not equal to $S$. Similarly, $\cup_{j \notin S_i} S_j$ is non-empty (or else $S_i$ would equal $S$), and is also not equal to $S$ (because it does not contain $i$). Because $|\cup_{j \in S \setminus S_i} S_j| < |S|$, it is covered. We have just written $S$ as the union of two covered sets, so therefore $S$ is also covered, a contradiction.
\end{proof}

Now, we are ready for the last step. We will argue that for all $g \in \mathcal{G}(f)$, knowledge of $g(S), g(T)$, and $g(S \cap T)$ suffices to deduce $g(S \cup T)$. We will then deduce that knowledge of $g(S)$ for all $S \in \mathcal{B}$ suffices to deduce $g(S)$ for all $S \in \mathcal{M}_f$.

\begin{lemma}\label{lem:cover}Let $S, T \in \mathcal{M}_f$. Then $\vec{v}^{S \cup T}=\vec{v}^S + \vec{v}^T - \vec{v}^{S \cap T}$.
\end{lemma}
\begin{proof}
Recall from the definition of weight-based that $h(S \cap T) \overline{\cup} h(S \cup T) \subseteq h(S) \overline{\cup} h(T)$. But recall also that $\sum_{i \in h(S \cap T) \overline{\cup} h(S \cup T)} w_i = \sum_{ i \in h(S) \overline{\cup} h(T)} w_i$ because all of $S, T, S \cap T, S \cup T$ are minimizers. As all $w_i$ are non-negative, this means that the only possible $i$ which are counted fewer times on the LHS than the RHS must have $w_i = 0$. This immediately means that $\vec{v}^{S \cap T} + \vec{v}^{S \cup T} = \vec{v}^S + \vec{v}^T$.
\end{proof}

\begin{corollary}\label{cor:main}Every $S \in \mathcal{M}_f$ has $\vec{v}^S \in \text{span}(\{\vec{v}^T\}_{T \in \mathcal{B}})$.
\end{corollary}
\begin{proof}
Lemma~\ref{lem:cover} establishes that if $S, T$, and $S \cap T$ are in $\text{span}(\{\vec{v}^T\}_{T \in \mathcal{B}})$, then so is $S \cup T$. Assume for contradiction that some $S \in \mathcal{M}_f$, $\vec{v}^S$ is not in $\text{span}(\{\vec{v}^T\}_{T \in \mathcal{B}})$, and take the one of minimal $|S|$. Then $S$ is covered by $\mathcal{B}$ by Lemma~\ref{lem:covered}, so we can write $S = A \cup B$, where $|A|,|B| < |S|$. $\vec{v}^A,\vec{v}^B$, and $\vec{v}^{A \cap B}$ are therefore in $\text{span}(\{\vec{v}^T\}_{T \in \mathcal{B}})$. By Lemma~\ref{lem:cover}, so then is $\vec{v}^S$, a contradiction.
\end{proof}

The proof is now concluded: we have argued that $|\mathcal{B}| \leq n+1$, and Corollary~\ref{cor:main} establishes that all of $\mathcal{M}_f$ is in the span of $\mathcal{B}$, so the Generalized Cut Dimension is at most $n+1$.

Importantly, observe that this proof fails to hold for the Generalized Non-Trivial Cut Dimension (it must, as we previously demonstrated an example with Generalized Non-Trivial Cut Dimension $3n/2-2$). The point of failure is that the set of Non-Trivial minimizers is not closed under intersection or union (if either the intersection is empty or the union is $[n]$). 
\end{proof}

\section{A $\binom{n}{2}$ Query Lower Bound for Non-Trivial SFM}\label{sec:nontriv}
In this section, we establish our lower bound for Non-Trivial SFM. The class of functions we consider will be the following:

\begin{definition}
A function $f(\cdot)$ is \emph{cost-based} if there exists a cost function $c:2^{[n]} \rightarrow \mathbb{R}_+$ with $c(T)=0$ whenever $|T| \leq 1$ such that $f(S) = \sum_{i \in S} f(\{i\}) - \sum_{T \subseteq S} c(T)$.
\end{definition} 

\begin{prop}
Every cost-based function is submodular.
\end{prop} 

\begin{proof}
We will establish that $f(S \cup \{i\}) - f(S) \leq f(T \cup \{i\}) - f(T)$ whenever $T \subseteq S$ and $i \notin S$. Observe that for any $X \subseteq [n]$ with $i\notin X$ we have $f(X \cup \{i\}) - f(X) = f(\{i\}) - \sum_{U \subseteq X} c(U \cup \{ i \})$.  $f(\{i\})$ is independent of $X$, and the second term is clearly at least as large for $X=S$ than $X=T$ (as $c(U) \geq 0$ for every $U$). Therefore, $f(\cdot)$ has diminishing marginal returns and is submodular.
\end{proof}

\begin{theorem}The query complexity of Non-Trivial SFM is at least $\binom{n}{2}$.

\end{theorem}
\begin{proof}
Consider the following $\binom{n}{2}+1$ cost functions. Define $c( \cdot )$ so that $c(S) = 0$ if $|S| \leq n-2$, $c([n]\setminus \{i\}) = n-1$ for all $i$, and $c([n]) = 2n$. Call the associated function $f(\cdot)$ where we set $f(\{i\}) = 1$ for all $i$. Then $f(S) = |S|$ when $|S| \leq n-2$, $f([n] \setminus \{i\}) = 0$, and $f([n]) = -n^2$. 

For every $1 \le i < j \le n$ define $c_{ij}(\cdot)$ so that $c_{ij}(S) = 0$ if $|S| \leq n-3$. For sets of size $n-2$, set $c_{ij}([n] \setminus \{i,j\}) = n-1$, and $c_{ij}([n] \setminus \{k,\ell\}) = 0$ for all other $\{k,\ell\} \neq \{i,j\}$. For sets of size $n-1$, set $c_{ij}([n] \setminus \{i\}) = c_{ij}([n] \setminus \{j\}) = 0$, and $c_{ij}([n]\setminus \{k\}) = n-1$ for all $k \notin \{i,j\}$. Finally, set $c_{ij}([n]) = 3n-1$. Call the associated function $f_{ij}(\cdot)$ where we set $f_{ij}(\{\ell\}) = 1$ for all $\ell$. Observe that $f_{ij}(S) = |S|$ when $|S| \leq n-3$, $f_{ij}(S) = n-2$ if $|S| = n-2$ and $S \neq [n] \setminus \{i,j\}$, $f_{ij}([n] \setminus \{i,j\}) = -1$,  $f_{ij}([n] \setminus \{\ell\}) = 0$ for all $\ell$, and $f_{ij}([n]) = -n^2$.

Observe, importantly, that the non-trivial minimum of $f(\cdot)$ is $0$, while the non-trivial minimum of $f_{ij}(\cdot)$ is $-1$ for all $i,j$. Observe also that $f(\cdot)$ and $f_{ij}(\cdot)$ \emph{differ only on their evaluation for $[n]\setminus \{i,j\}$}. Therefore, an adversary could answer any query $S$ with $f(S)$. If the algorithm terminates with fewer than $\binom{n}{2}$ queries, then there is some $[n] \setminus \{i,j\}$ that has not been queried. Therefore, the adversary is free to decide that the function is either $f(\cdot)$ or $f_{ij}(\cdot)$. As the value of the minima for these two functions are distinct, the algorithm cannot be correct. Therefore, any correct algorithm for Non-Trivial SFM must make at least $\binom{n}{2}$ queries.
\end{proof}
\section{Conclusions and Open Questions}
We establish the first query lower bounds exceeding $n$ for SFM ($2n$), Non-Trivial Symmetric SFM ($3n/2-2$), and Non-Trivial SFM ($\binom{n}{2}$). Our asymmetric lower bounds are from direct constructions. Our symmetric lower bound arises from the novel \emph{cut dimension}. 

Our work leaves open a clear direction for future work: what is the maximum possible Global Cut Dimension for an undirected graph? Or more generally, what is the maximum possible Non-Trivial Symmetric Generalized Cut Dimension of a weight-based function? %Our work establishes a graph with a $3n/2-2$ Global Cut Dimension, but our current upper bound is just $\binom{n}{2}$ (because the vectors are in $\mathbb{R}^{\binom{n}{2}}$). 

It is also of course generally important to further improve query complexity lower bounds for SFM variants (and also develop better algorithms). 

\bibliographystyle{alpha}
\bibliography{MasterBib}
\appendix

\section{Omitted Proofs from Section~\ref{sec:main}}\label{app:proofs}

We will make use of the following technical lemma:

\begin{lemma} If $R_i \subseteq S$, then $f_\sigma^{\vec{c}}(S) \le (|S|-i) \cdot (n + 2 - i)$. \end{lemma} 

\begin{proof}
If $S = R_{j(S)}$, then $f_\sigma^{\vec{c}}(S) = -c_j \le 0 \le (|S|-i) \cdot (n + 2 - i)$, as desired. Otherwise because $j(S) \ge i$ (by definition of $j(S)$), we know $|S|-j(S) \le |S| - i$ and $n + 2 - j(S) \le n + 2 - i$, which implies that $f_\sigma^{\vec{c}}(S) \le (|S|-i)\cdot  (n + 2 - i)$ as desired.\end{proof}

\begin{proof}[Proof of Lemma~\ref{lem:submodSFM}]

Let $X$, $Y$ be any two subsets of $[n]$; we will show that $$f_{\sigma}^{\vec{c}}(X) + f_{\sigma}^{\vec{c}}(Y) \ge f_{\sigma}^{\vec{c}}(X \cup Y) + f_{\sigma}^{\vec{c}}(X \cap Y).$$ Note that if $X \subseteq Y$, the inequality is trivially satisfied, as $X \cap Y = X$ and $X \cup Y = Y$; the inequality is also trivially satisfied if $Y \subseteq X$.  Hence, we will assume that neither set is contained in the other; note that this means neither set could equal $\emptyset$ or $[n]$. From here we consider two separate cases. 

In the first case, assume neither $X$ nor $Y$ is an important set. Let $R_i$ be the largest important set that is a subset of $X$ and $R_j$ be the largest important set that is a subset of $Y$. Without loss of generality, let's assume that $i \ge j$. Let $A := (X \setminus R_i) \setminus (Y \setminus R_i)$, $B := (Y \setminus R_i) \setminus (X \setminus R_i)$, $C := (X \setminus R_i) \cap (Y \setminus R_i)$, and $D := Y \cap (R_i \setminus R_j)$ and for convenience define $a := |A|$, $b := |B|$, $c := |C|$, and $d := |D|$. From the definition of $f$ we have that $$ f_{\sigma}^{\vec{c}}(X) = (a + c) \cdot (n + 2 - i) $$ and  $$f_{\sigma}^{\vec{c}}(Y) = (b + c + d) \cdot (n + 2 - j).$$ First, we prove the inequality when $i = j$. In this case, $D = \emptyset$, but $A$ and $B$ are both non-empty (as otherwise either $X \subseteq Y$ or $Y \subseteq X$). As $X \cap Y$ contains $R_i$ we have $$ f_{\sigma}^{\vec{c}}(X \cap Y) \le c (n + 2 - i)$$ from Lemma A.1. As $X \cup Y$ contains $R_i$ we similarly have $$ f_{\sigma}^{\vec{c}}(X \cup Y) \le (a + b + c) \cdot (n + 2 - i).$$ Hence, $$ f_{\sigma}^{\vec{c}}(X \cup Y) + f_{\sigma}^{\vec{c}}(X \cap Y) \le (a + b + 2c) (n + 2 - i) = (a + c) (n + 2 - i) + (b + c) (n + 2 - i) = f_{\sigma}^{\vec{c}}(X) + f_{\sigma}^{\vec{c}}(Y)$$ which proves the inequality as $d=0$. We'll next prove the inequality assuming $i > j$. In that case, we again have that $R_j$ is a subset of $X \cap Y$ so $$ f_{\sigma}^{\vec{c}}(X \cap Y) \le ( c + d) (n + 2 - j)$$ by Lemma A.1. Similarly, $R_i$ is a subset of $X \cup Y$ so $$ f_{\sigma}^{\vec{c}}(X \cup Y) \le (a + b + c) (n + 2 - i).$$ Hence it is sufficient to prove $$ (a + c) \cdot (n + 2 - i) + (b + c + d) \cdot (n + 2 - j) \ge (c+d) (n + 2 - j) + (a + b + c) (n + 2 - i)$$ which reduces to $b (n + 2 - j) \ge b (n + 2 - i)$, which is true as $b \ge 0$ and $i > j$. Hence we have completed our analysis for the first case.

Our second case is when $X$ is an important set and $Y$ is not. Say $X = R_i$ for some $i \ge 1$ and let $R_j$ be the largest important set contained in $Y$. Clearly, $i > j$ (as otherwise we would have $X \subseteq Y$). Let $A := Y \setminus R_i$, $B := (Y \setminus R_j) \cap R_i$, and for convenience define $a := |A|$ and $b := |B|$. Note that $a > 0$ (as otherwise $Y \subseteq X$) and that there are exactly $a + b$ elements of $Y$ not in $R_j$. Because $X \cap Y$ contains $R_j$ but not $R_{j + 1}$ we know $$ f_{\sigma}^{\vec{c}}(X \cap Y) = b \cdot (n + 2 - j). $$ Also, as $X \cup Y$ contains $R_i$ we know that $$ f_{\sigma}^{\vec{c}}(X \cup Y) \le a \cdot (n + 2 - i).$$ Hence it is sufficient to prove $$ - c_i + (a + b) \cdot (n + 2 - j) \ge b \cdot (n + 2 - j) + a \cdot (n + 2 - i).$$ This inequality is equivalent to $- c_i  + a (i - j) \ge 0$, which is true since $a \ge 1, i - j \ge 1$ and $c_i \le 1$, hence concluding the second case.

By symmetry, we now also have the same conclusion if $Y$ is an important set and $X$ is not. Moreover, we need not consider the case where both are important sets as then one must be a subset of the other. Hence we have the result. \end{proof}

\section{On Cut Queries in Directed Graphs}\label{app:auxiliary}

In this section, we examine the limits of cut queries when learning the edges of a graph. This appendix is not directly relevant to our main results, but may be of interest for future work.
\begin{claim}
A directed weighted graph can be learned via cut queries up to directed cycles. That is, with cut queries one can learn a graph $G'$ that is equivalent to the true graph $G$ up to adding/deleting directed cycles. Moreover, no set of queries can determine the weight of a directed cycle.
\end{claim}
\begin{proof}
We first work over unweighted graphs and show that a graph can be learned up to the direction of directed cycles. Let $G(V,E)$ be a directed, unweighted graph. First, we note that for any two vertices $u$ and $v$ in $V$, we can learn if
\begin{itemize}
    \item both edges $(u,v)$ and $(v,u)$ exist,
    \item exactly one of the edges $(u,v)$ and $(v,u)$ exist (but not which one), 
    \item or neither of these edges exist. 
\end{itemize}

To see this, consider creating a new function $f'(\cdot)$ which outputs $f(S) + f(V \setminus S)$. Then $f'(\cdot)$ corresponds to an undirected \emph{weighted} graph $G'$ where the weight between two nodes is zero, one, or two (and equal to the number of edges between them in $G$). As $G'$ is undirected, we can learn $G'$ exactly using cut queries~\cite{RubinsteinSW18}.\footnote{To see this, observe that $f'(\{u\}) + f'(\{v\}) - f'(\{u,v\})$ is exactly twice the weight of the edge between $u$ and $v$ in $G'$.}

Moreover, for every vertex $u$, by querying $\{u\}$ and $V \setminus\{u\}$, we know the in degree and the out degree of $u$. 

Now, suppose two graphs $G$ and $G'$ both satisfy all the queries made so far. We claim that $G$ can be converted to $G'$ by flipping the direction of certain cycles. Let $e = (u,v)$ be an edge in $G$ that does not exist in $G'$. We know that $(v,u)$ must be an edge in $G'$. Suppose we flip edge $(v,u)$ in $G'$. Now, the in degree of $u$ in $G'$ is one less than the in degree of $u$ in $G'$, while the out degree of $u$ in $G'$ is one more than that of $G$. Hence, there is an edge $(u,v')$ in $G'$, and an edge $(v',u)$ in $G$. We flip this edge in $G'$. Continuing this way, we flip all edges in a path, until we reach $v$. If $G'$ is the same as $G$, we are done, else, we pick another edge and repeat this procedure. Hence, $G'$ and $G$ are the same, up to the directions of directed cycles. 

Moreover, for any cut queried, every directed cycle either does not contribute anything to the cut or adds exactly 1 to the cut (irrespective of the direction). Hence, the direction of cycles in a directed graph cannot be learned by cut queries. 

This argument can be extended to weighted graphs as well. For weighted graphs, we can learn the sum of the weights of edges $(u,v)$ and $(v,u)$ for all edges, as well as the in degree and out degree of every vertex. 

Suppose we have two graphs $G$ and $G'$ which satisfy all the queries made to learn the above. Similar to the unweighted case, we claim that $G'$ can be converted to $G$ by changing the weights of certain directed cycles. Let $(u,v)$ be an edge which has weight $w$ in $G$ and $w'$ in $G'$. We change the weight of $(u,v)$ in $G'$ to $w$, and add $w'-w$ to the weight of $(v,u)$ in $G'$. Now, the in degree of $u$ has increased by $w' - w$. Since the in degree of $u$ is the same in both $G$ and $G'$ initially, this implies that there are edges incident on $u$ in $G$ whose weights differ from those in $G'$ by exactly $w'-w$. We change the weights of each of those edges to match the ones in $G$, continuing till we complete each of these cycles. Hence, $G$ and $G'$ are identical up to the weights of certain directed cycles. 

Note that each time we decrease the weight of a cycle in one direction by $\alpha$, we increase the weight of the cycle in the other direction by exactly $\alpha$. That is, the sum of the weights of the cycle in both directions remains constant. Let's assume that a cut query cuts $k$ edges of this cycle. This implies that the cut query also cuts $k$ edges of the cycle in the opposite direction. Hence, we can only determine the sum of the weights of these two cycles, irrespective of the number of queries made. 

\end{proof}

\begin{claim}
For a weighted undirected graph, when making $s$-$t$ cut queries, the weight of edge $(s, u)$ cannot be learned for any vertex $u$. Similarly, the weight of edge $(u,t)$ cannot be learned for any vertex $u$.
\end{claim}

\begin{proof}
We first note that we can always learn the weight of edges of the form $(u,v)$, with $u,v \not = s \not = t$. This can be done by querying $\{s\}$, $\{s,u\}$, $\{s,v\}$ and $\{s,u,v\}$. Hence, we can learn all edges except for those of the form $(s,u)$ and $(u,t)$ (because these queries are redundant or invalid in that case). After learning these weights, every query $S \cup \{s\}$ can be viewed as the sum of $n$ weights (the weights of edges from $S$ to $t$, and the weights of edges from $V\setminus S$ to $s$). Let us denote the weight of edge $(s,u)$ as $w_u$, and the weight of edge $(u,t)$ as $w'_u$. Every query $S \cup \{s\}$ can be written as a linear equation 
\[
\sum_{u \in S}w'_u + \sum_{u \in V\setminus S} w_u = c_S
\]
Let $\alpha_S$ denote a $2n$ dimensional vector with the coefficients of the above equation. Let $\mathbf{w}$ be a $2n$ dimensional vector with $w_u$ and the $w'_u$, for all $u$. The above equation can be written as 
\[
\langle \alpha_S, \mathbf{w} \rangle = c_S
\]
Let us consider the subspace $\Pi$ spanned by $\{ \alpha_S \colon S \subseteq V \setminus \{s,t\}\}$. Let $e_u$ denote the vector with a 1 in the position of edge $(s,u)$ and zeros elsewhere. If $e_u \in \Pi$, we can compute the value of $w_u$. We show that $e_u \not \in \Pi$, for all $u$. 

To show this, it is enough to describe a vector in the kernel of $\Pi$, whose dot product with $e_u$ is non-zero. Consider the vector $\beta$ with 1 in the position of edges $(s,u)$ and $(u,t)$, and $-\frac{1}{n}$ elsewhere. 
\[
\langle e_u, \beta \rangle = 1
\]
However, for any $S$, 
\[
\langle \alpha_S, \beta \rangle = 0
\]
Hence, we cannot compute the weight of edge $(s,u)$ for any vertex $u$. The same argument can also be made for the edge $(u,t)$. \end{proof}
\end{document}